\documentclass[conference]{IEEEtran}
\usepackage{graphicx}
\usepackage{amsmath}
\usepackage{floatrow}
\usepackage[T1]{fontenc}
\usepackage{amsthm}
\usepackage[noadjust]{cite}
\usepackage{filecontents}
\usepackage{epstopdf}
\usepackage{cite}
\usepackage{floatrow}
\usepackage{fix-cm}

\makeatletter
\newcommand{\Rmnum}[1]{\expandafter\@slowromancap\romannumeral #1@}
\makeatother

\theoremstyle{definition}
\newtheorem{defn}{Definition}[section]

\newtheorem{prop}{Proposition}
\date{}
\begin{document}
\title{{\fontsize{20}{20}\selectfont An Information-theoretic Model for Knowledge Sharing in Opportunistic Social Networks} \thanks{This work was funded in part by a Google Faculty Research Award.}}

\author{\large Mai ElSherief$^*$, Tamer ElBatt$^{\dagger\diamondsuit}$, Ahmed Zahran$^{\dagger^\diamondsuit}$, Ahmed Helmy$^{\ddagger}$  \\ [.1in]
\small  \begin{tabular}{c} 
$^*$Department of Computer Science, University of California, Santa Barbara, USA.\\
$^\dagger$Wireless Intelligent Networks Center (WINC), Nile University, Giza, Egypt.\\
$^\diamondsuit$Faculty of Engineering, Cairo University, Giza, Egypt.\\
$^\ddagger$The Department of Computer and Information Science and Engineering,
University of Florida, Gainesville, USA. \\
\end{tabular} }

\maketitle

\begin{abstract}
In this paper we establish fundamental limits on the performance of knowledge sharing in opportunistic social networks. In particular, we introduce a novel information-theoretic model to characterize the performance limits of knowledge sharing policies. Towards this objective, we first introduce the notions of knowledge gain and its upper bound, knowledge gain limit, per user. Second, we characterize these quantities for a number of network topologies and sharing policies. This work constitutes a first step towards defining and characterizing the performance limits and trade-offs associated with knowledge sharing in opportunistic social networks. Finally, we present numerical results characterizing the cumulative knowledge gain over time and its upper bound, using publicly available smartphone data. The results confirm the key role of the proposed 
model to motivate future research in this ripe area of research as well as new knowledge 
sharing policies.
\end{abstract}
\IEEEpeerreviewmaketitle
\vspace{-0.2 cm}

\section{Introduction}
The recent surge in mobile devices, complemented by a plethora of wireless 
communication standards, have inspired novel networking paradigms and services.
However, fully understanding and exploiting the social structure of 
mobile users remains a daunting challenge. Earlier
social studies, e.g., Homophily {[}Lazarsfeld and Merton (1954){]},
have shown that people tend to have similarities with others in close
proximity. 

The wide proliferation of resource-rich smartphones renders them tightly coupled to their users, bearing a wealth of behavioral data, e.g., locations, social networks, online shopping, etc., which infer information about the user's preferences and interests. Thus, there has been growing interest in leveraging this data to open new frontiers and enrich the user's life experiences \cite{Eagle}. An instance of this interaction also prevails in crowd sourcing applications which may affect the real-time user behavior, e.g., Waze and Google maps provide indicators for traffic congestion and road accidents which advise the mobile users to alter their routes.

Inspired by the tight coupling between smartphones and users' behaviors, we pose the following fundamental question: Can we capitalize on the wealth of knowledge and life experiences of people we encounter throughout our lives and may have common interests, yet we do not know? The proposed model caters to this question via an envisioned class of applications, coined {\it opportunistic recommendation systems} (ORS), where users capitalize on others' knowledge based on their mere co-existence and backed by homophily. The utility of ORS stems from extending our classic day-to-day ``physical'' recommendation exchanges, from people we know and encounter throughout the day, to ``virtual'' exchanges with users we opportunistically encounter and do not know (yet may have things in common according to homophily) and even to people we have never encountered, through the concept of knowledge sharing/forwarding, explored in this paper.

It is worth noting that similarity-based opportunistic social networks could serve as the basis for a variety of novel mobile services, e.g., trust establishment, targeted advertisement, friend finders, and location/similarity-based services. Furthermore, ORS is expected to spur a plethora of novel smartphone applications serving large public venues, e.g., museums, theme parks, shopping malls, sports events and fairgrounds.
%

In this paper, we establish fundamental limits for opportunistic social 
networks, particularly focusing on knowledge sharing among ``similar'' users. 
In \cite{mai14}, we addressed the problem of assessing pair-wise user similarity
in mobile societies, as a pre-requisite for knowledge sharing. We examined a 
number of classic similarity metrics, e.g., cosine, as well as novel ones 
for non-temporal (vector) and temporal (matrix) user profiles. Our major findings 
in \cite{mai14} assert that temporal metrics are generally more strict in declaring similarity, 
compared to non-temporal ones. Second, vectorized metrics, e.g., the introduced 
vectorized-cosine, constitute a low-complexity approach towards realizing the, typically, more complex temporal similarity on mobile devices. Finally, information-theoretic metrics, e.g., Hellinger distance, hold great promise for quantifying similarity between probability distribution user profiles. 

In \cite{mobihoc}, the authors study the problem of content dissemination in 
opportunistic social networks. Their main result shows that high contact rate, 
non-social nodes (i.e. rarely found in ``temporal communities'') are mostly 
responsible for efficient content dissemination. However, unlike this work, the 
model adopted in \cite{mobihoc} is not information-theoretic. 

Infrormation-theoretric models have been employed in other areas, e.g., cooperative data compression and distributed source coding for data gathering in multi-hop wireless sensor networks (WSNs) with spatial correlations, e.g., \cite{ramachandran,servetto,marco,govindan,elbatt}. However, the prime focus of this research thrust of research is to eliminate redundancies 
among the possbily correlated sensor measurements to transport sensor data with no/minimal 
redundnacy using minimal energy and bandwidth resources \cite{govindan,servetto}. The joint entropy of the random variables representing the individual sensors as the sources of data constitutes the lower bound on the traffic volume generated by the sensors, where source coding algorithms try to achieve. On the other hand, our objective in this work is fundamentally different. We establish fundamental limits, using basic information theoretic constructs, on the maximum knowledge avaiable for a user to reap in a given opportunistic encounter of similar users. As defined later in the sequel, the knowledge gain limit of an arbitrary user constitutes the upper bound on the amount of knowledge a user can reap in a given opportunistic encounter and is characerized by the joint entropy of the random variables modling the indivudal knowledge each user bears in diverse areas of life. Furthermore, the nodes are stationary in data gathering WSNs and communications is multi-hop, whereas in our problem setting nodes are generally mobile and communications is limited to single-hop, with the possibility of forwarding the knowledge acquired from previous encounters.

Our main contribution in this paper is a novel information-theoretic
model for knowledge sharing in opportunistic social networks. 
First, we introduce the new notions of {\it Knowledge Gain} and 
its upper bound {\it Knowledge Gain Limit}, per user.
Second, we establish fundamental limits and unveil key insights for 
diverse network topologies and sharing policies and validate 
our theoretical findings using publicly available smartphone user traces
under plausible scenarios. The use of modeling abstractions to study formation, 
dynamics and evolution of social networks is not new. For instance, graph theory has been employed extensively 
in social networks to model patterns of networks, clustering, homophily and basic concepts like centrality and connectedness, e.g., \cite{jackson,graph1,graph3}. In addition, random graph theory has been employed to model social network formation, evolution and growth, e.g., \cite{erdos,bollobas,wasserman,graph2}, among other topics. However, to the best of the authors' knowledge, employing information-theoretic tools to model knowledge sharing in opportunistic, mobile social networks has not been explored before.

Our prime focus in this paper is to introduce the information-theoretic model, establish fundamental limits, as opposed to designing and implementing specific knowledge sharing schemes, which constitute an interesting topic of future research. This work lays out a 
mathematical foundation for analyzing and contrasting existing and future 
knowledge sharing and delay-tolerant forwarding policies in opportunistic social networks.
%

The rest of this paper is organized as follows. In Section II, we introduce 
the system model and underlying assumptions. In Section III, we present
some preliminaries. Afterwards, we introduce the information-theoretic model 
and quantify the performance of candidate knowledge sharing policies in Section IV. 
In Section V, we present key results, based on realistic smartphone data, confirming 
our theoretical findings. Finally, conclusions are drawn and potential directions for 
future research are pointed out in Section VI.
%
\vspace{-0.1 cm}
\section{System Model}
%
%
\vspace{-0.1 cm}  
We model an opportunistic encounter of $M$ similar users as a wireless ad hoc network where 
all nodes are peers and wish to establish communications to share ``knowledge'' and recommendations in various life categories. To study the problem at hand, we assume that all users trust each other and that each user is pair-wise similar to all other users in the sense of \cite{mai14}, as an example. Each user has its own profile vector, $\bar{P}$, modeled as a probability distribution across $v$ generic life categories, e.g., arts, sports, shopping, etc., chosen by the profile designer based on target application(s), e.g., \cite{odp}. Each user is assumed to have a table of recommendations that stores \textit{tips} (knowledge) for sharing with other similar users, e.g., upcoming event(s), best sellers, site visits, etc. We assume that the users leverage short-range wireless communications with fixed transmission power (i.e. circular disk model), e.g., WiFi Direct or Bluetooth. Hence, interference and medium access issues are assumed resolved using these, or other, protocols and lie out of the scope of this work.

In this paper, we wish to address two fundamental questions pertaining to knowledge sharing:\\
{\bf 1.} For an arbitrary user $i$, what is the maximum amount of knowledge available for this user (knowledge limit) in a given similarity-based opportunistic encounter (network)?\\
{\bf 2.} For user $i$, what is the amount of knowledge gain that is achievable, i.e. the user can reap from similar users in an opportunistic social network, using a specific knowledge sharing policy?
\vspace{-0.2 cm}
\subsection{Knowledge Limit and Knowledge Gain}
In this section, we introduce two key concepts that are fundamental to the proposed model, namely the knowledge limit and knowledge gain.
\begin{defn}
The Knowledge Gain Limit ($KL_i$) is defined, for an arbitrary user $i$, as the maximum amount of knowledge that is available for user $i$ to extract from similar users in a given network.
\end{defn}
\begin{defn}
The Knowledge Gain ($KG_i$) is defined, for an arbitrary user $i$, as the amount of knowledge user $i$ can gain from similar users in a given network, using a specific knowledge sharing policy.
\end{defn}
\vspace{-0.2 cm}
It is straightforward to notice that $KG_i \le KL_i$ since the knowledge limit constitutes the upper bound on the knowledge that can be reaped, irrespective of the sharing policy. Inspired by the probability distribution definition of the user profile, we argue that probability- and information-theoretic tools would prove useful for modeling and analyzing the system at hand.

In the next section, we introduce the formal definition of the knowledge gain per encounter. For modeling convenience, we assume that the tips a user has follow the same probability distribution $\bar{P}$ (i.e. have the same mix) as the profile for that user. This does not only facilitate the mathematical analysis but can also be justified by the observation that users, typically, tend to have more tips in life categories they are more interested in, as reflected by their profile vectors.
%

%
\vspace{-0.1 cm}
\section{Preliminaries}
\subsubsection{The Knowledge Gain per encounter}

We recall from information theory that the Entropy of a discrete-valued random variable $X$, denoted $H(X)$, represents a measure of the ``uncertainty'' in this random variable which also represents the amount of information this random variable bears \cite{cover}. Given our assumption that the user tips follow the same probability distribution as the user profile, user tips can also be modeled as a discrete-valued random variable, $X$. Accordingly, $H(X)$ 
quantifies the amount of information (or knowledge)\footnote{We use the terms Knowledge and Information interchangebly in this paper.} this user has. This simple model opens room for quantifying the newly introduced concepts of knowledge limit and gain.

To better illustrate the concepts, we first consider a toy example of an ``opportunistic encounter'' that involves only two users within the wireless communication range of each other. The two users have tips probability distribution vectors, denoted $X$ and $Y$. Assume users $X$ and $Y$ opportunistically meet and are deemed similar\footnote{We abuse notation and use tips PMFs, $X$ and $Y$, to refer also to the users.}, possibly using metrics in \cite{mai14} and, hence, should exchange ``knowledge'', that is, informative tips. Based on simple entropy relationships, 
we distinguish three types of tips that are key to our discussion:\\
%
1. Tips that user $X$ has, but $Y$ does not: given by $H(X|Y)$.\\
2. Tips that user $Y$ has, but $X$ does not: given by $H(Y|X)$.\\
3. Tips that both users have: given by $I(X;Y)$, where $I(X;Y)$ is the mutual information between $X$ and $Y$.

Thus, the knowledge gained by user $X$ from $Y$ can be defined as
I\begin{equation}
KG(X) = H (Y|X) = H(X,Y) - H(X)
\end{equation}
where $H(X,Y)$ is the joint entropy of the two random variables representing the users tips vectors.
%
%
The third type of tips (common to both users), characterized as the mutual information $I(X;Y)$, constitutes the ``communication overhead'' since it is exchanged despite the fact that it does not contribute to increasing the knowledge of $X$ or $Y$. This perfectly 
agrees with our assumption that the two users know nothing about each other, when they meet opportunistically for the first time and, hence, this overhead is unavoidable.
%

It is worth noting that the knowledge gain limit for user $X$ (or $Y$), in this simple example of two users, is equal to the knowledge gain. 

\subsubsection{The Knowledge Gain Limit}
Based on the information-theoretic definitions of the KG and KL for two users established in the previous section, we generalize the definition to characterize the knowledge limit for user $X_1$, without loss of generality, in an opportunistic encounter with $M-1$ other users deemed similar to $X_1$, as follows: 
\vspace{-0.1 cm}
\begin{equation}
		KL(X_1) = H(X_1, X_2, X_3, ......, X_M) - H(X_1)
\end{equation}
%
which can be written as
\vspace{-0.1 cm}
\begin{equation}
KL(X_1) = H(X_2|X_1) + H(X_3|X_2,X_1) + .....
\nonumber
\end{equation}
\begin{equation}
\hspace *{-0.15 cm} + H(X_M|X_{M-1}, ......, X_1).
\end{equation}

\vspace{-0.2 cm}
Thus, (2) asserts that the maximum amount of knowledge that user $X_1$ can extract from the network is simply the joint entropy (knowledge) that all users have, after removing any redundant knowledge, which is represented by the joint entropy, $H(X_1, X_2, X_3, ......, X_M)$, less the amount of information that user $X_1$ already has, that is, $H(X_1)$. It is worth noting that the KL characterization in (2) and (3) is general enough, valid for all network topologies and is independent of any knowledge sharing policy.


\vspace{-0.2 cm}
\section{Knowledge Sharing in Opportunistic Social Networks: Fundamental Limits}

We now utilize the basic definitions introduced in the previous section to establish the KL for an arbitrary user in diverse scenarios as well as characterize its KG, under two candidate knowledge sharing policies:\\
{\bf 1.} Send my tips only, or ``{\it Send Mine Only}'', whereby a user sends only own 
tips to a similar, directly encountered user.\\ 
{\bf 2.} Forward my tips and others, or ``{\it Forward Mine Plus Others}'', whereby a user forwards his/her own tips along with those acquired so far from others in previous encounters.

It is worth noting that these two policies are mere examples to illustrate the concept, however, other policies can be introduced and analyzed using the proposed model.
For instance, a user may forward his/her own tips along with a subset of others' tips based on some criteria. This gives rise to a family of knowledge sharing policies that deserve a comprehensive analysis, to assess their merits and potential trade-offs, which lies out of the scope of this work. 
%

Next, we shift our attention to quantifying the KL and KG achievable by the two aforementioned knowledge sharing policies, under a variety of opportunistic social network settings. In particular, we consider two topology scenarios, namely all nodes are within the communication range of each other (i.e. single-hop scenario) and multi-hop scenarios, where some nodes may lie out of the communication range of others. In this paper, we limit our attention to fixed topology encounters (i.e. stationary or quasi-stationary users) and leave user mobility for future work.
%
%

%
\vspace{-0.2 cm}
%
%
\subsection{Single-hop Networks}
Under this setting, the users may be stationary, or quasi-stationary, yet, any
node remains always one-hop away from all other nodes. Thus, the nodes' movement does not alter the network topology, which always remains fully connected. 
For this setting, we can easily characterize the knowledge gain limit, as in (3), and, further, prove that the knowledge gain will always achieve the limit, under loose delay constraints. This is in complete agreement with intuition since any node can take turns to exchange tips with all other nodes directly reachable. Thus, ``all'' knowledge that is available for a user in this network, can be fully reaped. 
%
The achievability result for the {\it Send Mine Only} (SMO) policy is established by the following proposition.
%
%

\begin{prop}
For single-hop networks, an arbitrary node can achieve its knowledge gain limit using the {\it SMO} policy.
\end{prop}
\vspace{-0.1 cm}
\begin{proof}
Without loss of generality, we assume that node $X_1$ encounters other nodes in an ascending order of their IDs. Under the {\it Send Mine Only} policy, 
the cumulative knowledge gain for node $X_1$, $KG(X_1)$, after receiving tips from all other nodes $X_2, X_3, X_4,...., X_M$ in turn, is given by $H(X_2|X_1) + H(X_3|X_2,X_1) + .....+ H(X_M|X_{M-1}, ......, X_1)$, which is the same as the $KL(X_1)$ in (3). 
The same argument can be applied to all other nodes in the network which proves the result.
%
%
\end{proof}


%
%
As indicated earlier, one of the fundamental questions in our study is how long does it take a user to attain the knowledge limit, if it is at all attainable. This is directly related to the number of exchanges that a user performs to attain the KL, as confirmed by simulations which reveal insteresting insights discussed later in Section V.A. Under the {\it Send Mine Only} policy and assuming that each node has at least one unique tip to contribute to the knowledge gain limit, it is straightforward to show that the worst-case number of exchanges needed for an arbitrary node to attain the KL is $M-1$, that is, $O(M)$.

In the rest of this section, we shift our attention to quantify the KG and time-to-achievability of the {\it Forward Mine Plus Others} (FMPO) sharing policy, in single-hop networks. Thus, a user shares not only its own tips but also tips collected from previous encounters, denoted by the subscript $(.)_p$. We prove in the following proposition that the knowledge gain limit is also achievable using the FMPO policy under loose delay constraints.

%

\begin{prop}
For single-hop networks, an arbitrary node achieves the knowledge gain limit using the FMPO policy.
\end{prop}
\vspace{-0.2 cm}
\begin{proof}
We proceed along the lines of Proposition 1 and give an outline of the proof due to space limitations. Without loss of generality, we assume that node $X_1$ encounters all other nodes in an ascending order of their IDs. Thus, $KG(X_1)$ based on encountering nodes $X_2, X_3, X_4,...., X_M$ in turn, is given by
\begin{equation}
KG(X_1) = H(X_2, \vec{X_{2p}}|X_1) + H(X_3, \vec{X_{3p}}|X_2, \vec{X_{2p}}, X_1) 
\nonumber
\end{equation} 
\begin{equation}
\hspace*{0.9 cm} + H(X_4, \vec{X_{4p}}|X_3, \vec{X_{3p}}, X_2, \vec{X_{2p}}, X_1)+ .....
\nonumber
\end{equation}
\begin{equation}
\hspace*{1.3 cm} + H(X_M, \vec{X_{Mp}}|X_{M-1}, \vec{X_{(M-1)p}} ,......, X_1),
\end{equation}
\noindent where $\vec{X_{ip}}$ captures the previous encounters of node $X_i$. It can 
be shown that each joint entropy term in (4) can be expanded, 
e.g., $H(X_2, \vec{X_{2p}}|X_1)$ becomes $H(X_2|X_1) + H(\vec{X_{2p}}|X_2, X_1)$, where 
the second term (previous encounter tips of a node) would be redundant in some cases (acquired from earlier encounters) and, hence, contributes zero to the KG. This reduces (4) to the KL in (3) and proves the result.
\end{proof}
It should be noted that once the conditioning, in the conditional entropy terms in the RHS, accommodates all nodes in the network, the incremental gain becomes zero and the node achieves its knowledge limit. In essence, the role of previous encounters (appearing in the conditional entropy terms) is the sole contributor to the FMPO policy attaining the KL faster than SMO, which will be shown in Section V. Apparently, this does not come for free since there is a fundamental trade-off between the cumulative KG after a number of encounters and the associated communication overhead which warrants attention in future research, especially in multi-hop networks. We prove next that the communication overhead of FMPO is greater than or equal to SMO, in single-hop networks.

\begin{prop}
For single-hop networks, the communication overhead under FMPO is greater than or equal to SMO.
\end{prop}
\vspace{-0.2 cm}
\begin{proof}
Assume two users, $X$ and $Y$, encounter each other. Denote the vector of previous encounters for $X$ and $Y$ by $\vec{X_p}$ and $\vec{Y_p}$, respectively.

Under the SMO policy, the communication overhead that $X$ incurs is the common knowledge (mutual information) between what $X$ sends (which is its knowledge only) and $Y$'s
knowledge so far which is given by 
\begin{equation}
OH(X)_{SMO}=I(X;Y,\vec{Y_p}).
\label{overheadX}
\end{equation}
Similarly, the communication overhead from the perspective of user $Y$ is $OH(Y)_{SMO}=I(Y;X,\vec{X_p})$.

Under FMPO, the communication overhead is the same for both users, $X$ and $Y$, and is given by
\begin{equation} 
OH(X)_{FMPO}=OH(Y)_{FMPO}=I(X,\vec{X_p};Y, \vec{Y_p}).
\label{overheadXY}
\end{equation}

From information theory, the mutual information between two random variables $A$ and $B$ can be written as
\begin{equation}
I(A;B)=H(A) + H(B) - H(A,B).
\label{mutInfoEqn}
\end{equation}
Applying ~\eqref{mutInfoEqn} on ~\eqref{overheadX} yields
\begin{equation}
OH(X)_{SMO}= H(X)+H(Y,\vec{Y_p})-H(X,Y,\vec{Y_p}).
\label{Eqn.4}
\end{equation}
Applying ~\eqref{mutInfoEqn} on ~\eqref{overheadXY} yields
\begin{equation}
\hspace*{-1 cm} OH(X)_{FMPO}=OH(Y)_{FMPO}= H(X,\vec{X_p})
\nonumber
\end{equation}
\begin{equation}
\hspace*{1.3 cm} +H(Y,\vec{Y_p})-H(X,\vec{X_p},Y,\vec{Y_p}).
\label{Eqn.3}
\end{equation}
Subtracting ~\eqref{Eqn.4} from ~\eqref{Eqn.3} yields 
\begin{equation}
\hspace*{-1 cm} OH(X)_{FMPO}-OH(X)_{SMO}=H(X,\vec{X_p})-H(X)
\nonumber
\end{equation}
\begin{equation}
\hspace*{3 cm} -H(X,\vec{X_p},Y,\vec{Y_p})+H(X,Y,\vec{Y_p}).
\label{Eqn.5}
\end{equation}
Since $H(A,B) = H(A) + H(B|A) = H(B) + H(A|B) $, ~\eqref{Eqn.5} can be re-written as
\begin{equation}
OH(X)_{FMPO}-OH(X)_{SMO}=H(X)+H(\vec{X_p}|X)-H(X)
\nonumber
\end{equation}
\begin{equation}
-[H(\vec{X_p}|X,Y,\vec{Y_p})+H(X,Y,\vec{P_p})] + H(X,Y,\vec{Y_p})
\nonumber
\end{equation}
which reduces to 
\begin{equation}
OH(X)_{FMPO}-OH(X)_{SMO}=H(\vec{X_p}|X)- H(\vec{X_p}|X,Y,\vec{Y_p})
\label{lasteqn}
\end{equation}
\begin{equation}
\hspace*{0.6 cm} \ge 0
\label{ineq}
\end{equation}
where the inequality in~\eqref{ineq} follows since conditioning reduces entropy. This proves the result.
\end{proof}
\vspace{-0.2 cm}
\subsection{Fixed Topology Multi-hop Networks}
Under this setting, we assume the network topology is connected and time-invariant where some nodes are not directly reachable from each other, i.e. multi-hop paths will be followed for exchanging 
some tips. In this case, the role of the knowledge sharing policy stands out and 
affects whether a user can/cannot attain the KL.
%
%

Next, we quantify the performance, and trade-offs, of the two aforementioned knowledge sharing policies, namely SMO and FMPO. In case of
SMO, the knowledge gain achieved by node $X_1$ is limited by the neighborhood size, $N<M$. This renders the KG under this policy strictly less than the KL for multi-hop networks, even under loose delay constraints. The following proposition establishes this result.

\begin{prop}
For fixed topology, multi-hop networks, the {\it Send Mine Only} knowledge sharing policy will not attain the knowledge gain limit, that is, $KG(X_1)<KL(X_1)$ iff $N < M$.
\end{prop}
%
\vspace{-0.2 cm}
\begin{proof}
Without loss of generality, we assume that node $X_1$ communicates with direct neighbors in an ascending order of their IDs. The cumulative knowledge gain for $X_1$, $KG(X_1)$, according to communications with neighbors $X_2, X_3, X_4,...., X_{N}$ is given by $H(X_2|X_1) + H(X_3|X_2,X_1) + .....+ H(X_{N}|X_{N-1}, ......, X_1)$. It is worth noting that the summation of positive conditional entropy terms is limited to $N<M$ nodes. It misses other positive terms involving the $M-N$ non-neighbors of $X_1$. Hence, it directly follows that $KG(X_1)<KL(X_1)$, which proves the result.
%
\end{proof}
%

%
%


Next, we focus on the performance of FMPO for fixed topology multi-hop networks. As expected, forwarding others tips opens room for any node to achieve its KL under loose delay constraints, even when $N<M$. The following proposition formally establishes this result. 
%

\begin{prop}
For fixed topology, multi-hop networks, an arbitrary node can achieve the knowledge gain limit using the FMPO knowledge sharing policy.
\end{prop}
%
\begin{proof}
We provide an outline of the proof due to space limitations. Without loss of generality, we assume that each node starts off with its own knowledge only and $X_1$ encounters its single-hop neighbors in an ascending order of their IDs, while other neighbors have pair-wise encounters with other nodes in the network. The cumulative knowledge gain for node $X_1$, $KG(X_1)$, after meeting neighboring nodes $X_2, X_3, X_4,...., X_N$ is given by
\begin{equation}
KG(X_1)=H(X_2,|X_1) + H(X_3, \vec{X_{3p}}|X_2,X_1) 
\nonumber
\end{equation}
\begin{equation}
\hspace*{1.4 cm} + H(X_4, \vec{X_{4p}}|X_3, \vec{X_{3p}}, X_2,X_1)+.....
\nonumber
\end{equation}
\begin{equation}
\hspace*{1.3 cm} + H(X_N, \vec{X_{Np}}|\vec{X_{N-1}}, \vec{X_{(N-1)p}} ,......, X_1). 
\end{equation}
At this point, two cases arise. First, if the previous knowledge vectors, $\vec{X_{ip}}~\forall i$, bear the ``forwarded'' tips from all non-neighboring nodes, namely $X_{N+1}, X_{N+2}, ......, X_M$, then it can be shown that the cumulative knowledge gain of $X_1$ becomes
\begin{equation}
KG(X_1)=H(X_1, X_2, X_3, ......, X_M) - H(X_1)= KL(X_1) 
\end{equation} 
\noindent which proves the result. On the other hand, if the previous encounters do not cover knowledge from all non-neighbors, then this implies that node $X_1$ needs more time to attain $KL(X_1)$. Backed by network connectedness and loose delay constraints, $X_1$ can attain its KL almost surely via recurring pairing with neighbors that it has already paired with until it acquires all missing knowledge from nodes out of its communication range. 
\end{proof}

\vspace{-0.3 cm}
\section{Performance Results}

In this section, we back our theoretical findings with numerical results based on 
smartphone user profile traces \cite{data}.
\vspace{-0.4 cm}
\subsection{Single-hop Networks}
\vspace{-0.1 cm}
%
%
In this section, we rely on real user behavior traces. We utilize traces of user interests in 24 life categories for $20$ smartphone users from the LiveLab project \cite{data}. The users are assumed stationary, or quasi-stationary, and, hence, the network topology is a full mesh. In order to quantify the knowledge gain and its limit for an arbitrary user, we pre-process the huge amount of data as follows. First, we compute the joint probability mass function for the 20 users' activities over a period of six months, from September 2010 to February 2011. Towards this end, we monitor the users' activities categorized under the $24$ categories
each second, for six months and record their concurrent activities. Afterwards, we divide by the total duration of the six months to get the joint PMF. Next, we show the cumulative knowledge gain increase as the node under investigation encounters more nodes over time.

First, we present performance results for a single-hop network under the 
SMO policy. For the network of $M=20$ nodes discussed earlier, an 
arbitrary user can achieve the knowledge gain limit within $M-1=19$ encounters. This is shown in Fig.~\ref{fig:B00_SSHOP)} for three arbitrary users, namely $B00$, $B04$ and $D03$. It can be noticed that the cumulative knowledge gain is a non-decreasing function with time. On the other hand, the knowledge gain limit is shown as a horizontal solid line that is generally different from one user to another.
\begin{figure}[!tp]
\includegraphics[width=10cm ,height=5.2cm]{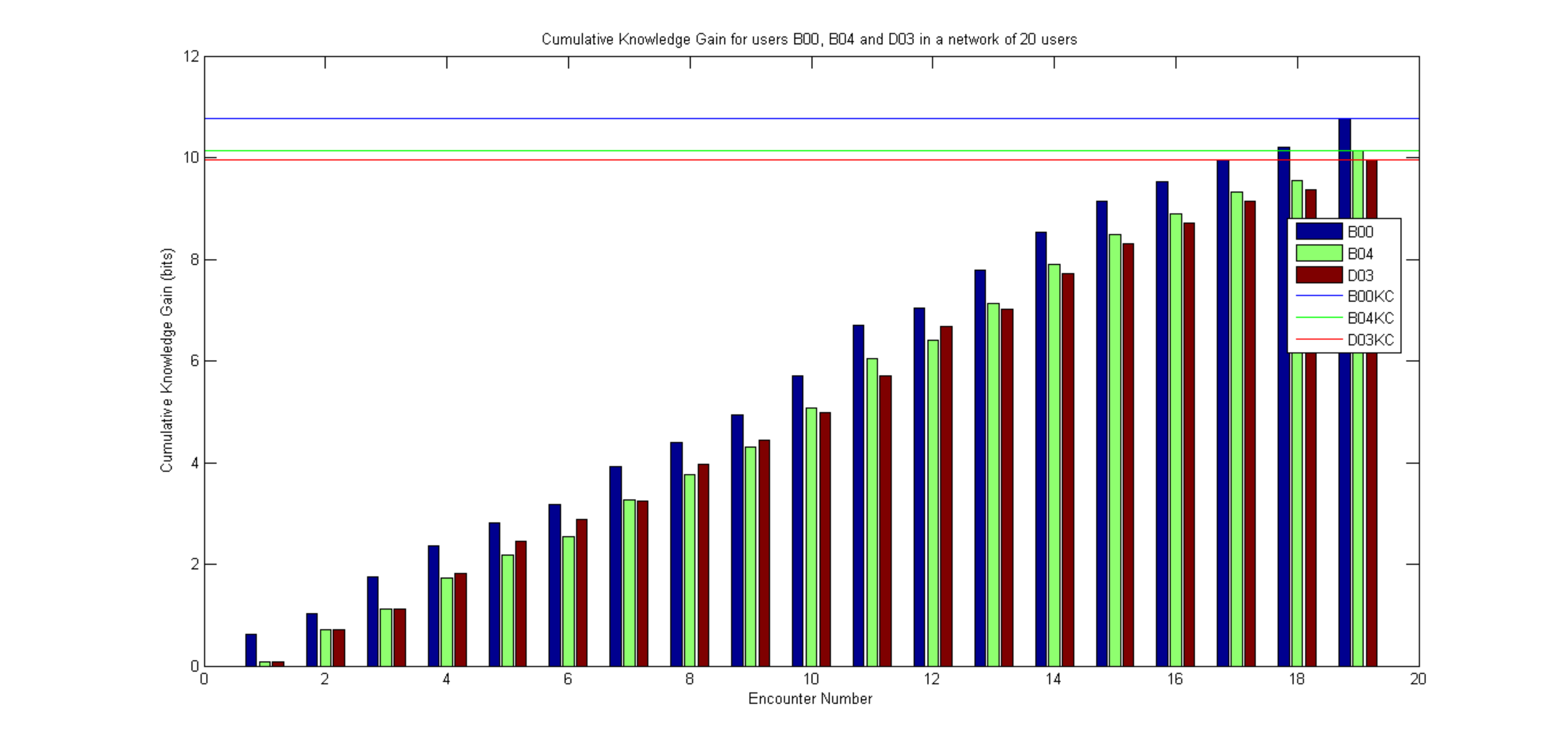}
    \caption{Cumulative knowledge gain for three users in a single-hop network under SMO.}\label{fig:B00_SSHOP)}
\end{figure}
%

Next, we shift our attention to the same network setting, yet, employing the FMPO policy.
Based on Proposition 2, we show that, for single-hop networks, all nodes achieve the knowledge gain limit using FMPO, yet, faster than SMO, i.e. in less encounters due to sharing the tips of others. This valuable insight is confirmed for the same three users, $B00$, $B04$ and $D03$, in Fig.~\ref{fig:B00_SSHOP(MO)}.
\begin{figure}[!bp]
  \centerline{  \includegraphics[width=10cm ,height=5.2cm]{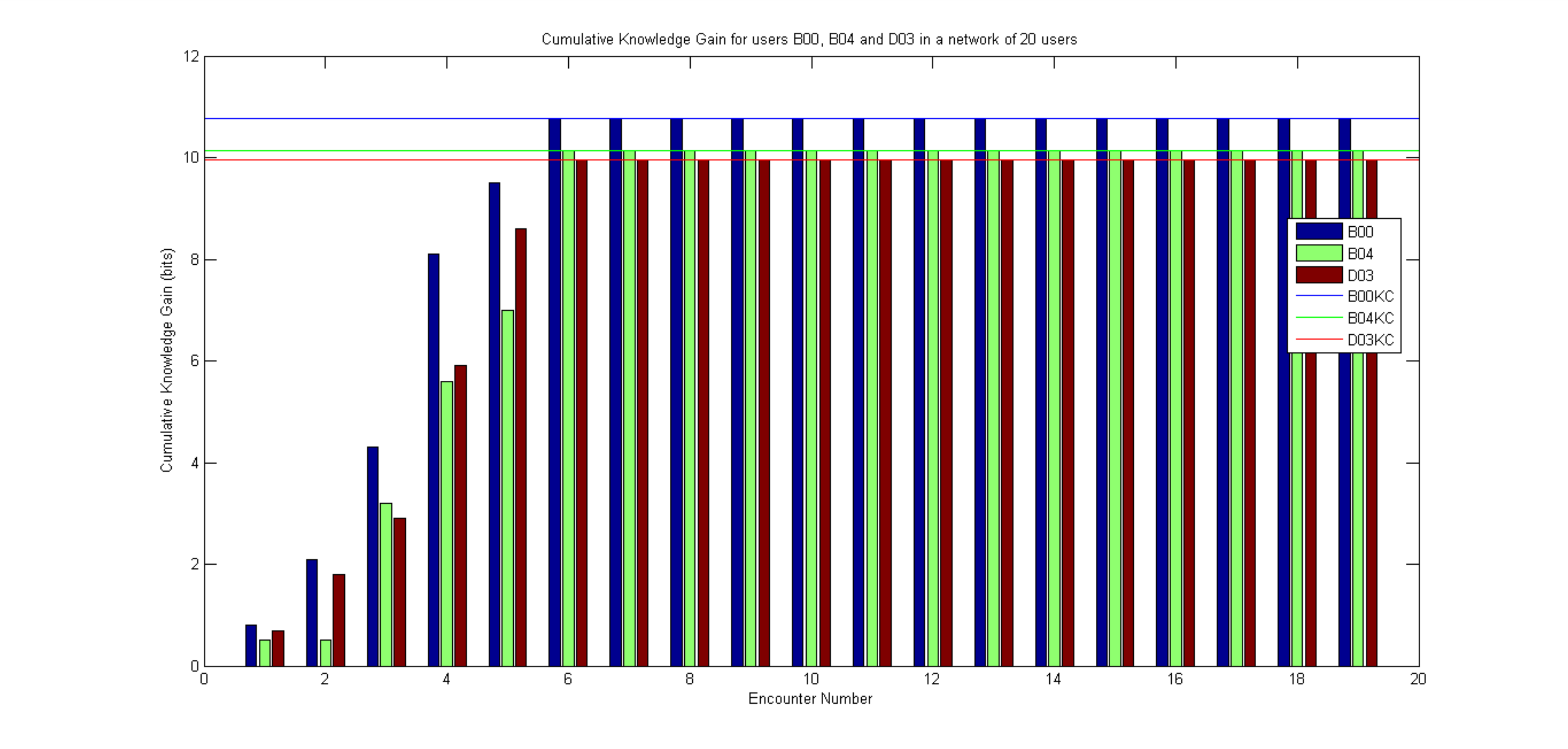}}
    \caption{Cumulative knowledge gain for three users in a single-hop network under FMPO.}\label{fig:B00_SSHOP(MO)}
\end{figure}
\vspace{-0.1 cm}
\subsection{Fixed Topology Multi-hop Networks}
\vspace{-0.1 cm}

In this section, we shift our attention to fixed topology multi-hop networks 
where we study the cumulative knowledge gain behavior under the SMO and FMPO policies.

As indicated earlier, and shown formally in Proposition 4, achieving the KL of an arbitrary node using SMO is fundamentally limited by the single-hop neighborhood size of this node, denoted $N$.
To this end, we generate $20$ randomly generated topologies of uniformly distributed users whereby each user has 6-7 single-hop neighbors, out of $20$ nodes, on the average. It can be noticed that $B00$ 
does not achieve the KL available for it in this network, as established in Proposition 4 and shown here using real smartphone user behavior traces in Fig. ~\ref{fig:B00_SMHOP)}. Thus, the maximum KG node $B00$ can achieve is only 43\% of its KL. Similarly, users $B06$ and $D00$ have single-hop neighbors strictly less than $M=20$ and, hence, cannot achieve their respective KLs. 
%
\begin{figure}[!bp]
 \centerline{\includegraphics[width=10cm ,height=5.2cm]{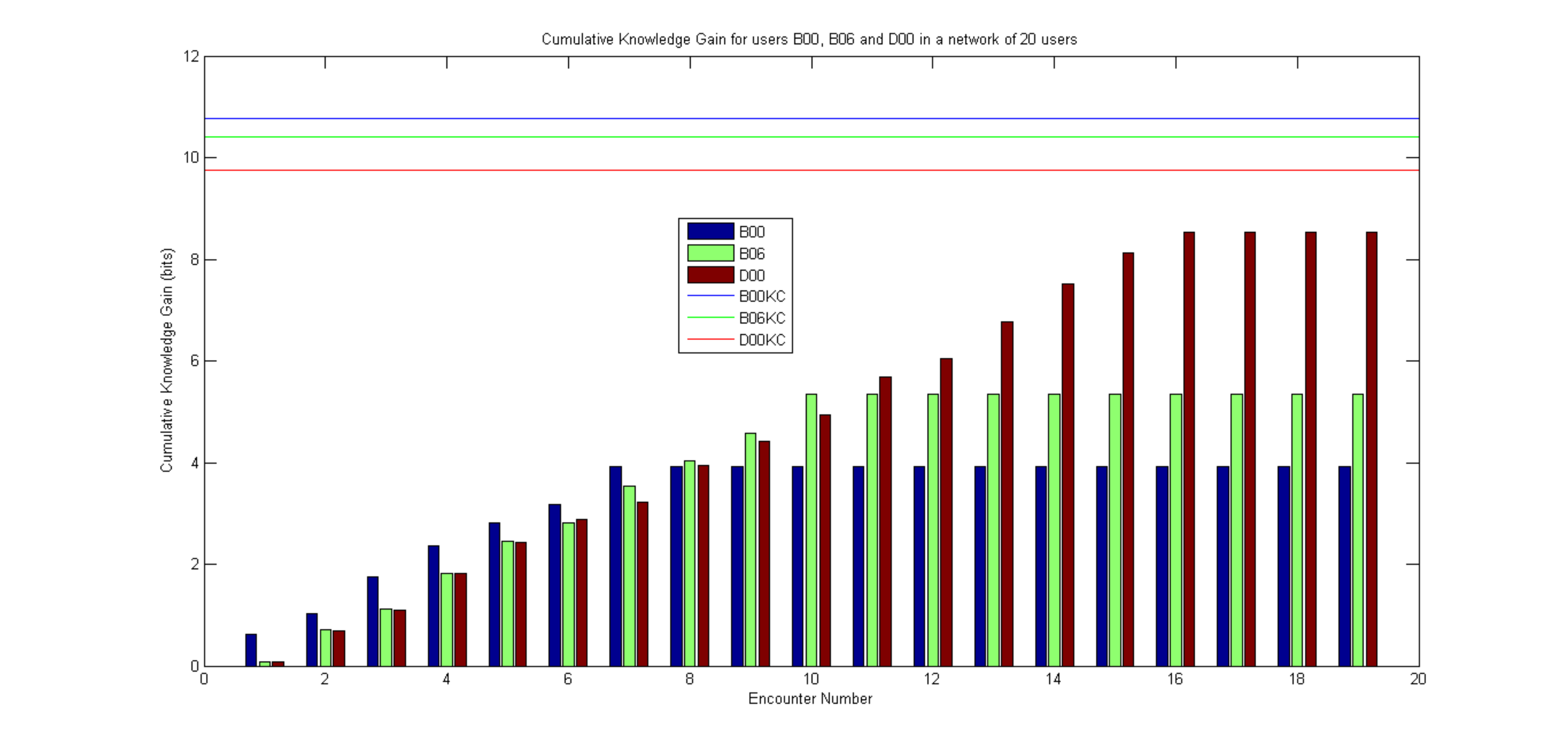}}
    \caption{Cumulative knowledge gain for three users in a fixed topology multi-hop network under SMO.}\label{fig:B00_SMHOP)}
\end{figure} 
    

Finally, we have shown that FMPO overcomes the limited neighborhood problem in multi-hop scenarios due to forwarding others' knowledge and, hence, nodes achieve the KL, as proven in Proposition 5 and shown in Fig. 4. 

%
  \begin{figure}[!bp]
\centerline  {   \includegraphics[width=10cm ,height=5.2cm]{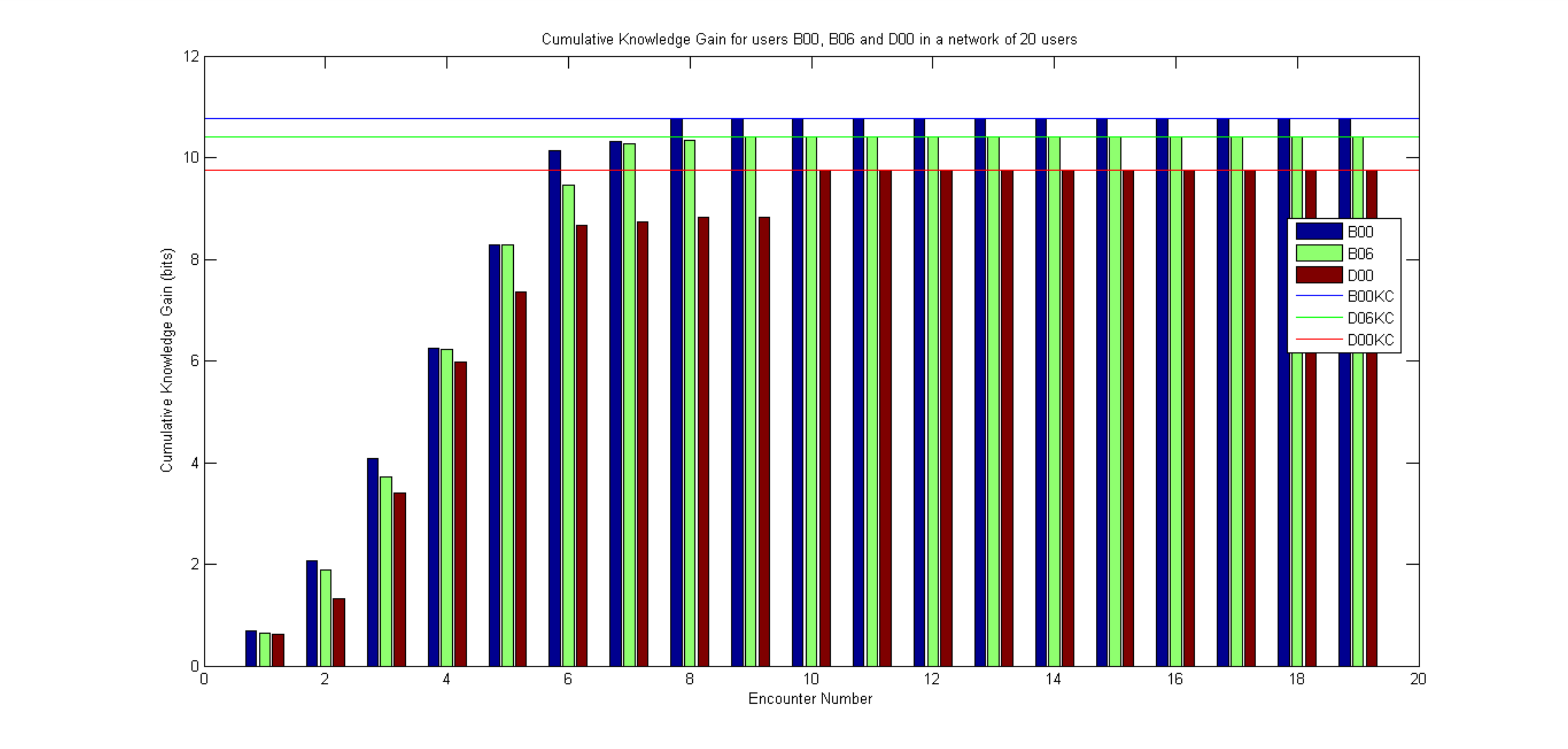}}
    \caption{Cumulative knowledge gain for three users in a fixed topology multi-hop network under FMPO.}\label{fig:B00_SMHOP_MO)}
\end{figure}

\vspace{-0.3 cm}
\section{Conclusion}
In this paper we propose a novel information-theoretic model for knowledge sharing in opportunistic social networks. We introduce the new notions of knowledge gain and its upper bound, knowledge gain limit, and establish fundamental limits and insightful results for two knowledge sharing policies among users meeting opportunistically. We present numerical results characterizing the knowledge gain limit for different users and the cumulative knowledge gain over time, using publicly available traces for smartphone user behavior from the LiveLab project, for fixed single- and multi-hop topologies. This work can 
be extended along multiple directions, e.g., propose efficient knowledge sharing policies and leverage the proposed mathematical model to analyze them and establish fundamental limits for mobile multi-hop network scenarios.




\vspace{-0.2 cm}
{\footnotesize {
\bibliography{KG_References}
\bibliographystyle{IEEEtran}
}}

\end{document}